\documentclass[conference, 10pt, final]{IEEEtran}

\usepackage{enumitem}
\usepackage[nolist]{acronym} 
\usepackage{tikz}
\usepackage{pgfplots}
\usepgfplotslibrary{groupplots}
\usepackage{graphicx}
\usepackage{subcaption}
\usepackage{float}
\pgfplotsset{compat=newest}
\usepackage{hyperref}
\usepackage{units}
\usepackage{amsmath, amsbsy, amssymb, amsthm}
\usepackage{tikzscale}
\usetikzlibrary{arrows}
\usepackage{color}
\usepackage{epigraph}
\usepackage{circuitikz}
\usepackage{multirow}
\usepackage{booktabs}
\usepackage[plainruled]{algorithm2e}
\usepackage[group-separator={,}]{siunitx}
\definecolor{darkgreen}{rgb}{0.125,0.5,0.169}
\usetikzlibrary{shapes,arrows}
\usetikzlibrary{positioning}
\usetikzlibrary{patterns}
\usetikzlibrary{decorations.pathreplacing}
\usepackage{marvosym}
\usepackage{bbm}
\usepackage{mathtools}
\usepackage{xfrac}
\usepackage{numprint}
\usepackage{bm}

\tikzset{>=latex}

\begin{acronym}
 \acro{CSI}{channel state information}
 \acro{UE}{user equipment}
 \acro{UL}{uplink}
 \acro{BS}{basestation}
 \acro{TDD}{time division duplex}
 \acro{FDD}{frequency division duplex}
 \acro{ECC}{error-correcting code}
 \acro{MLD}{maximum likelihood decoding}
 \acro{HDD}{hard decision decoding}
 \acro{IF}{intermediate frequency}
 \acro{RF}{radio frequency}
 \acro{SDD}{soft decision decoding}
 \acro{NND}{neural network decoding}
 \acro{CNN}{convolutional neural network}
 \acro{ML}{maximum likelihood}
 \acro{GPU}{graphical processing unit}
 \acro{BP}{belief propagation}
 \acro{LTE}{Long Term Evolution}
 \acro{BER}{bit error rate}
 \acro{SNR}{signal-to-noise-ratio}
 \acro{ReLU}{rectified linear unit}
 \acro{BPSK}{binary phase shift keying}
 \acro{QPSK}{quadrature phase shift keying}
 \acro{AWGN}{additive white Gaussian noise}
 \acro{MSE}{mean squared error}
 \acro{LLR}{log-likelihood ratio}
 \acro{MAP}{maximum a posteriori}
 \acro{NVE}{normalized validation error}
 \acro{BCE}{binary cross-entropy}
 \acro{CE}{cross-entropy}
 \acro{BLER}{block error rate}
 \acro{SQR}{signal-to-quantisation-noise-ratio}
 \acro{MIMO}{multiple-input multiple-output}
 \acro{OFDM}{orthogonal frequency division multiplex}
 \acro{RF}{radio frequency}
 \acro{LOS}{line of sight}
 \acro{NLoS}{non-line of sight}
 \acro{NMSE}{normalized mean squared error}
 \acro{CFO}{carrier frequency offset}
 \acro{SFO}{sampling frequency offset}
 \acro{IPS}{indoor positioning system}
 \acro{TRIPS}{time-reversal IPS}
 \acro{RSSI}{received signal strength indicator}
 \acro{MIMO}{multiple-input multiple-output}
 \acro{ENoB}{effective number of bits}
 \acro{AGC}{automatic gain control}
 \acro{ADC}{analog to digital converter}
 \acro{ADCs}{analog to digital converters}
 \acro{FB}{front bandpass}
 \acro{FPGA}{field programmable gate array}
 \acro{JSDM}{Joint Spatial Division and Multiplexing}
 \acro{NN}{neural network}
 \acro{IF}{intermediate frequency}
 \acro{LoS}{line-of-sight}
 \acro{NLoS}{non-line-of-sight}
 \acro{DSP}{digital signal processing}
 \acro{AFE}{analog front end}
 \acro{SQNR}{signal-to-quantisation-noise-ratio}
 \acro{SINR}{signal-to-interference-noise-ratio}
 \acro{ENoB}{effective number of bits}
 \acro{PCB}{printed circuit board}
 \acro{EVM}{error vector mangnitude}
 \acro{CDF}{cumulative distribution function}
 \acro{MRC}{maximum ratio combining}
 \acro{MRP}{maximum ratio precoding}
 \acro{MRT}{maximum ratio transmission}
 \acro{DeepL}{deep-learning}
 \acro{DL}{downlink}
 \acro{SISO}{single-input single-output}
 \acro{SGD}{stochastic gradient descent}
 \acro{CP}{cyclic prefix}
 \acro{MISO}{Multiple Input Single Output}
 \acro{LMMSE}{linear minimum mean square error}
 \acro{ZF}{zero forcing}
 \acro{USRP}{universal software radio peripheral}
 \acro{RNN}{recurrent neural network}
 \acro{GRU}{gated recurrent unit}
 \acro{LSTM}{long short-term memory}
 \acro{NTM}{neural turing machine}
 \acro{DNC}{differentiable neural computer}
 \acro{TCN}{temporal convolutional network}
 \acro{FCL}{fully connected layer}
 \acro{MANN}{memory augmented neural network}
 \acro{RNN}{recurrent neural network}
 \acro{DNN}{dense neural network}
 \acro{FIR}{finite impulse response}
 \acro{BPTT}{back-propagation through time}
 \acro{GAN}{generative adversarial network}
 \acro{ELU}{exponential linear unit}
 \acro{tanh}{hyperbolic tangent}
 \acro{BICM}{bit-interleaved coded modulation}
 \acro{OTA}{over-the-air}
 \acro{IM}{intensity modulation}
 \acro{DD}{direct detection}
 \acro{RL}{reinforcement learning}
 \acro{SDR}{software-defined radio}
 \acro{WGAN}{Wasserstein generative adversarial network}
 \acro{BMD}{bit-metric decoding}
 \acro{BMI}{bit-wise mutual information}
 \acro{LDPC}{low-density parity-check}
 \acro{IDD}{iterative demapping and decoding}
 \acro{JSD}{Jensen-Shannon divergence}
 \acro{MMSE}{minimum mean square error}
 \acro{FFT}{fast Fourier transform}
 \acro{IFFT}{inverse fast Fourier transform}
 \acro{QAM}{quadrature amplitude modulation}
 \acro{EMD}{earth mover's distance}
 \acro{TDL}{tapped delay line}
 \acro{KL}{Kullback-Leibler}
 \acro{PRACH}{physical random access channel}
 \acro{URLLC}{ultra-reliable low-latency communication}
 \acro{ANOMA}{asynchronous non-orthogonal multiple access}
 \acro{FEC}{forward error correction}
 \acro{PAPR}{peak-to-average power ratio}
 \acro{APP}{a posteriori probability}
 \acro{COTS}{commercial off-the-shelf}
 \acro{PLL}{phase locked loop}
 \acro{STO}{sampling time offset}
 \acro{SFO}{sampling frequency offset}
 \acro{CFO}{carrier frequency offset}
 \acro{CPO}{carrier phase offset}
 \acro{CSI}{channel state information}
 \acro{GNSS}{global navigation satellite system}
 \acro{ELAA}{extremely large aperture array}
 \acro{UE}{user equipment}
 \acro{DICHASUS}{\underline{Di}stributed \underline{Cha}nnel \underline{S}ounder by \underline{U}niversity of \underline{S}tuttgart}
 \acro{JCAS}{Joint Communication and Sensing}
 \acro{AoA}{angle of arrival}
\end{acronym}

\definecolor{mittelblau}{RGB}{0, 126, 198}
\definecolor{violettblau}{cmyk}{0.9, 0.6, 0, 0}
\definecolor{rot}{RGB}{238, 28 35}
\definecolor{apfelgruen}{RGB}{140, 198, 62}
\definecolor{gelb}{RGB}{1, 221, 0}
\definecolor{orange}{RGB}{244, 111, 33}
\definecolor{pink}{RGB}{237, 0, 140}
\definecolor{lila}{RGB}{128, 10, 145}
\definecolor{hellgrau}{RGB}{224, 224, 224}
\definecolor{mittelgrau}{RGB}{128, 128, 128}
\definecolor{dunkelgrau}{RGB}{80,80,80}
\definecolor{anthrazit}{RGB}{19, 31, 31}

\definecolor{bgorange}{HTML}{fcc0a7}
\definecolor{bggreen}{HTML}{ccebb9}

\newtheorem{theorem}{Theorem}

\DeclareMathOperator*{\argmax}{arg\,max}

\IEEEoverridecommandlockouts

\begin{document}

\title{Deep Learning for Uplink CSI-based Downlink Precoding in FDD massive MIMO\\Evaluated on Indoor Measurements}

\author{\IEEEauthorblockN{Florian Euchner, Niklas S\"uppel, Marc Gauger, Sebastian D\"orner, Stephan ten Brink \\}

\IEEEauthorblockA{
Institute of Telecommunications, Pfaffenwaldring 47, University of  Stuttgart, 70569 Stuttgart, Germany \\ \{euchner,sueppel,gauger,doerner,tenbrink\}@inue.uni-stuttgart.de
}

%\thanks{Thanks}
}

\maketitle

\begin{abstract}
    When operating massive \ac{MIMO} systems with \ac{UL} and \ac{DL} channels at different frequencies (\ac{FDD} operation), acquisition of \ac{CSI} for downlink precoding is a major challenge.
    Since, barring transceiver impairments, both \ac{UL} and \ac{DL} \ac{CSI} are determined by the physical environment surrounding transmitter and receiver, it stands to reason that, for a static environment, a mapping from \ac{UL} \ac{CSI} to \ac{DL} \ac{CSI} may exist.
    First, we propose to use various \ac{NN}-based approaches that learn this mapping and provide baselines using classical signal processing.
    Second, we introduce a scheme to evaluate the performance and quality of generalization of all approaches, distinguishing between known and previously unseen physical locations.
    Third, we evaluate all approaches on a real-world indoor dataset collected with a 32-antenna channel sounder.
\end{abstract}

\acresetall

\section{Introduction and Problem Statement}
Massive \ac{MIMO} is widely accepted to be a crucial technology for increasing the spectral efficiency of future cellular wireless systems through spatial multiplexing.
At the multi-antenna \ac{BS}, it relies on precoding in the \ac{DL} direction, which requires accurate \ac{CSI} for the channel between \ac{BS} and \ac{UE}.
The \ac{BS} estimates \ac{CSI} for the \ac{UL} channel from pilots transmitted by the \ac{UE}.
In \ac{TDD} operation, thanks to channel reciprocity, \ac{DL} \ac{CSI} can be directly derived from \ac{UL} \ac{CSI}.
If, however, \ac{UL} and \ac{DL} channels are at different frequencies (\ac{FDD} operation), acquisition of \ac{DL} \ac{CSI} is challenging.
Sending downlink pilots and obtaining \ac{CSI} feedback from \acp{UE} produces overhead that can become prohibitively large for high numbers of antennas \cite{yang2019deep}.

\begin{figure}
    \centering
    \scalebox{0.9}{
        \input{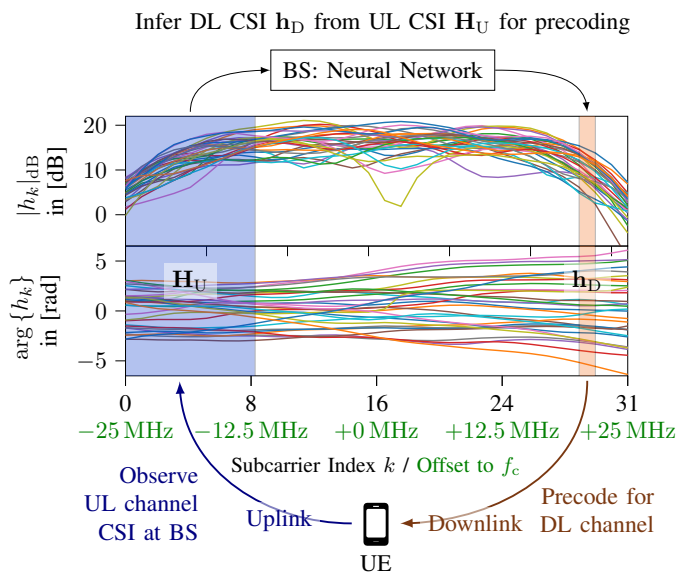}
    }
    \caption{Basic principle of operation}

    \label{fig:principle}
    \vspace{-0.5cm}
\end{figure}

Even as massive \ac{MIMO} was originally conceived, it was conjectured that \ac{DL} \ac{CSI} feedback in \ac{FDD} operation could be rendered unnecessary by exploiting relationships between \ac{UL} and \ac{DL} \ac{CSI} \cite[Section VII.J]{marzetta}.
For example, in typical radio environments, measurements have indicated that angles of arrival and departure are similar for \ac{UL} and \ac{DL} channels \cite{hugl2002spatial}.
However, in environments with many scatterers and potentially more than one strong propagation path, the relationship between \ac{UL} and \ac{DL} \ac{CSI} is no longer this simple.
Under the premise that the mapping from \ac{UL} \ac{CSI} to \ac{DL} \ac{CSI} is bijective, which is reasonable to assume for many practical environments \cite{alrabeiah2019deep}, a \ac{DNN} is capable of learning this mapping.
This learning-based approach, illustrated in Fig. \ref{fig:principle}, has been proposed in several earlier publications and has been evaluated on simulated channel models \cite{yang2019deep} \cite{zhang2020cv}.
By contrast, experiments with measured channel data are rare \cite{arnold2019enabling} and cannot be replicated without the underlying datasets.
We address these issues by making the following contributions:
\begin{itemize}
    \item We derive upper and lower bounds for downlink precoding performance in Section \ref{sec:baselines}.
    \item In Section \ref{sec:deeplearning}, we verify the concept of deep learning-based \ac{CSI} estimation, which has primarily been developed with simulated channels, on a measured, publicly available \ac{CSI} dataset, which is introduced in Section \ref{sec:dataset}.
    \item In Section \ref{sec:generalization}, we compare the quality of \ac{DL} \ac{CSI} estimates for different \ac{NN} architectures for our particular dataset and propose an evaluation framework for different network architectures that takes into account the difference in the quality of estimates in previously seen and unseen regions of the physical environment\footnote{A tutorial for a special case of \ac{DL} \ac{CSI} estimation is available at \mbox{\url{https://dichasus.inue.uni-stuttgart.de/tutorials/tutorial/downlinkcsi/}}}.
\end{itemize}

\section{Model, Metrics and Baselines}
\label{sec:baselines}
In the context of this work, we always consider the case of a single \ac{BS} antenna array with $M$ co-located antennas and a single \ac{UE} with one antenna.
We assume \ac{OFDM}-modulated signals for both uplink and downlink, but restrict ourselves to estimating the \ac{CSI} for a single subcarrier in the downlink channel.
It is important to note that this approach can easily be extended to all subcarriers in the downlink channel by using one estimator per subcarrier.

We denote the unknown channel coefficient vector for this particular \ac{DL} subcarrier by $\mathbf h_\mathrm{D} \in \mathbb C^M$.
We assume that \ac{UL} channel coefficients $\mathbf H_\mathrm{U} \in \mathbb C^{M \times N_\mathrm{sub}}$ for all antennas and all $N_\mathrm{sub}$ uplink subcarriers are known to the \ac{BS}.
We furthermore neglect hardware impairments and noise and assume that both \ac{UL} and \ac{DL} \ac{CSI} are determined by some latent variable $\mathbf x$, which captures all properties of the radio environment such as location and orientation of transmitter, receiver and scatterers:
\begin{equation}
    f_U: \mathbf x \mapsto \mathbf H_\mathrm{U} \quad \text{and} \quad f_D: \mathbf x \mapsto \mathbf h_\mathrm{D}
    \label{eq:envchannelmapping}
\end{equation}

In Eq. (\ref{eq:envchannelmapping}), $f_\mathrm{U}$ and $f_\mathrm{D}$ are deterministic mappings from environment properties $\mathbf x$ to \ac{UL} and \ac{DL} channel coefficients, respectively.
If $f_U$ is bijective, which has been argued to be probable in practical environments \cite{alrabeiah2019deep} \cite{vieira2017deep}, and the \ac{BS} is capable of learning $f_\mathrm{D} \circ f_\mathrm{U}^{-1}$, which \acp{NN} are theoretically capable of according to the universal approximation theorem, it can compute $\mathbf h_\mathrm{D}$ as $\mathbf h_\mathrm{D} = f_\mathrm{D} \circ f_\mathrm{U}^{-1}\left( \mathbf H_\mathrm{U} \right)$.
In practice, $f_\mathrm{U}$ may only be bijective on a (large) subset of the domain, the universal approximation theorem only holds for arbitrarily large \ac{NN} sizes, only limited training data is available and \ac{UL} channel estimates are noisy, hence the learned mapping $\hat { \bm \theta }: \mathbf H_\mathrm{U} \to \mathbf w$ will only produce an \emph{estimate} $\mathbf w \in \mathbb C^M$, $\mathbf w \approx \mathbf h_\mathrm{D}$ for the true downlink channel $\mathbf{h}_\mathrm{D}$.

A suitable metric for the quality of the estimate $\mathbf w$ for one particular realization of the channel is given by the squared cosine similarity of $\mathbf h_\mathrm{D}$ and $\mathbf w$:
\begin{equation}
    P = \frac{\left|\mathbf h_\mathrm{D}^\mathrm{H} \mathbf w\right|^2}{\left\lVert \mathbf h_\mathrm{D} \right\rVert^2 \lVert \mathbf w \rVert^2}
    \label{eq:power}
\end{equation}

In contrast to a \ac{MSE} metric, the expression for $P$ in Eq. (\ref{eq:power}) has the advantage of being interpretable as the normalized received power on the considered downlink subcarrier at the \ac{UE} when precoding with vector $\mathbf w^*$ and transmitting across the channel $\mathbf h_\mathrm{D}$.
In this sense, a normalized received power of $P = 1$ corresponds to perfect knowledge of the downlink channel down to a global phase rotation, i.e., $\mathbf w = \mathrm e^{\mathrm j \varphi} \mathbf h_\mathrm{D}$ with arbitrary $\varphi \in \mathbb R$.

Eq. (\ref{eq:power}) refers to one particular downlink channel $\mathbf h_\mathrm{D}$ and estimate $\mathbf w$.
In practice, $\mathbf w$ is estimated at the \ac{BS} based on $\mathbf H_\mathrm{U}$, i.e., $\mathbf w = \hat { \bm \theta } \left( \mathbf H_\mathrm{U} \right)$, and $\mathbf H_\mathrm{U}$ and $\mathbf h_\mathrm{D}$ are modelled as random variables that are jointly distributed over some distribution $\mathcal H$:
\[
    (\mathbf H_\mathrm{U}, \mathbf h_\mathrm{D}) \sim \mathcal H
\]

To obtain a more universal indicator $\bar P$ for the performance of a \ac{DL} \ac{CSI} estimator $\hat {\bm \theta}$, we consider the expected value of $P$, i.e., the average normalized received power, over the whole distribution $\mathcal H$:
\begin{equation}
    \bar P = \mathrm{E}_{(\mathbf H_\mathrm{U}, \mathbf h_\mathrm{D}) \sim \mathcal H} \left[\frac{\left|\mathbf h_\mathrm{D}^\mathrm{H} \hat { \bm \theta }(\mathbf H_\mathrm{U})\right|^2}{\left\lVert \mathbf h_\mathrm{D} \right\rVert^2 \lVert \hat { \bm \theta }(\mathbf H_\mathrm{U}) \rVert^2} \right]
    \label{eq:poweravg}
\end{equation}

Without any knowledge about $\mathcal H$, it is still possible to achieve an average normalized received power of $\bar P = \frac{1}{M}$ through the use of random precoding vectors, as the following theorem will show.

\begin{theorem}[Random Precoding Baseline]
    \label{thm:randomprecoding}
    For any arbitrary distribution $(\mathbf H_\mathrm{U}, \mathbf h_\mathrm{D}) \sim \mathcal H$, \textbf{random precoding} with channel estimates $\mathbf w \in \mathbb C^M, ~ \mathbf w := \frac{\mathbf v}{\lVert \mathbf v \rVert}$ where $\mathbf v \sim \mathcal {CN}(\mathbf 0, \mathbf I_M)$ and independent of $\mathbf h_\mathrm{D}$, $\mathbf H_\mathrm{U}$ leads to an expected received power
    \begin{equation}
        \bar P_\mathrm{rand} := \mathrm{E} \left[ \frac{|\mathbf h_\mathrm{D}^\mathrm{H} \mathbf w|^2}{\lVert \mathbf h_\mathrm{D} \rVert^2} \right] = \frac{1}{M}.
        \label{eq:randomprecoding}
    \end{equation}

\end{theorem}
\begin{proof}
    Noticing that $\mathbf h_\mathrm{D}$ and $\mathbf w$ are independent and since $|\mathbf h_\mathrm{D}^\mathrm{H} \mathbf w|^2 = \mathbf h_\mathrm{D}^\mathrm{H} \mathbf w \mathbf h_\mathrm{D}^\mathrm{T} \mathbf w^*$, we can exchange the order of the expectation operator and the scalar product sums in Eq. (\ref{eq:randomprecoding}):
    %their joint probability distribution can be separated into $p(\mathbf w, \mathbf h_\mathrm{D}) = p(\mathbf w) ~ p(\mathbf h_\mathrm{D})$, and with the definition of $\mathrm E[\cdot]$ and according to Fubini's theorem, we obtain
    %\[
    %    \bar P_\mathrm{rand} = \mathrm{E}_{\mathbf h_\mathrm{D}} \left[ %\mathrm{E}_{\mathbf w} \left[ \frac{|\mathbf h_\mathrm{D}^\mathrm{H} \mathbf w|^2}{\lVert \mathbf h_\mathrm{D} \rVert^2} \right] \right].
    %\]
    \begin{align*}
        \bar P_\mathrm{rand} %&= \mathrm{E}_{\mathbf h_\mathrm{D}} \left[ \frac{1}{\lVert \mathbf h_\mathrm{D} \rVert^2} ~ \mathrm{E}_{\mathbf w} \left[ \sum_{i = 1}^M \left( h_{\mathrm{D}, i}^* w_i \right) ~ \sum_{j = 1}^M \left( h_{\mathrm{D}, j} w_j^* \right) \right] \right] \\
        %&= \mathrm{E}_{\mathbf h_\mathrm{D}} \left[ \frac{1}{\lVert \mathbf h_\mathrm{D} \rVert^2} ~ \mathrm{E}_{\mathbf w} \left[ \sum_{i = 1}^M \sum_{j = 1}^M h_{\mathrm{D}, i}^* w_i  h_{\mathrm{D}, j} w_j^* \right] \right] \\
        &= \mathrm{E}_{\mathbf h_\mathrm{D}} \left[ \frac{1}{\lVert \mathbf h_\mathrm{D} \rVert^2} ~ \sum_{i = 1}^M \sum_{j = 1}^M h_{\mathrm{D}, i}^* h_{\mathrm{D}, j} \mathrm{E}_{\mathbf w} \left[w_i w_j^*\right] \right].
    \end{align*}
    
    Next, we need to show that $\mathrm E[w_i w_j^*] = 0$ for $i \neq j$.
    For this, first note that the distribution of $\mathbf v$ is invariant under unitary transformations $Q$, and so is $\mathbf w$ since $Q \mathbf w = \frac{Q \mathbf v}{\lVert \mathbf v \rVert} = \frac{Q \mathbf v}{\lVert Q \mathbf v \rVert}$.
    In particular, this implies that the distributions of $\mathbf w = (\ldots, w_i, \ldots, w_j, \ldots)$ and 
    $\mathbf w' = (\ldots, -w_i, \ldots, w_j, \ldots)$ are identical and hence $\mathrm{E} \left[ w_i w_j^* \right] = -\mathrm{E} \left[ w_i w_j^* \right] = 0$.

    With this, $\bar P_\mathrm{rand}$ further simplifies to
    \begin{align*}
        \bar P_\mathrm{rand} %&= \mathrm{E}_{\mathbf h_\mathrm{D}} \left[ \frac{1}{\lVert \mathbf h_\mathrm{D} \rVert^2} ~ \sum_{i = 1}^M h_{\mathrm{D}, i}^* h_{\mathrm{D}, i} ~ \mathrm{E}_{\mathbf w} \left[w_i w_i^*\right] \right] \\
        &= \mathrm{E}_{\mathbf h_\mathrm{D}} \left[ \frac{1}{\lVert \mathbf h_\mathrm{D} \rVert^2} ~ \sum_{i = 1}^M |h_{\mathrm{D}, i}|^2 ~ \mathrm{E}_{\mathbf w} \left[|w_i|^2\right] \right].
    \end{align*}
    For symmetry reasons, $\mathrm E_{\mathbf w} \left[|w_i|^2 \right] = \mathrm E_{\mathbf w} \left[|w_j|^2 \right]$ for any $i, j$.
    Hence, $\mathrm E_{\mathbf w} \left[\mathbf w^\mathrm{H} \mathbf w \right] = \mathrm E_{\mathbf w} \left[\sum_{i = 1}^M |w_i|^2 \right] = \sum_{i = 1}^M \mathrm E_{\mathbf w} \left[|w_i|^2 \right] = M ~ \mathrm E_{\mathbf w} \left[|w_i|^2 \right]$.
    Since $\mathbf w^\mathrm{H} \mathbf w = \lVert\mathbf w\rVert^2 = 1$, we find that $\mathrm E_{\mathbf w} \left[|w_i|^2\right] = \frac{1}{M}$ for any $i$:

    \begin{align*}
        \bar P_\mathrm{rand} &= \mathrm{E}_{\mathbf h_\mathrm{D}} \left[ \frac{\sum_{i = 1}^M h_{\mathrm{D}, i}^* h_{\mathrm{D}, i}}{\lVert \mathbf h_\mathrm{D} \rVert^2} ~ \frac{1}{M} \right]
    \end{align*}
    
    Since $\sum_{i = 1}^M h_{\mathrm{D}, i}^* h_{\mathrm{D}, i} = \lVert \mathbf h_\mathrm{D} \rVert^2$, all terms depending on $\mathbf h_\mathrm{D}$ cancel, which proves that Eq. (\ref{eq:randomprecoding}) holds for arbitrary $\mathcal H$:
    \begin{align*}
        \bar P_\mathrm{rand} &= \mathrm{E}_{\mathbf h_\mathrm{D}} \left[ \frac{1}{M} \right] = \frac{1}{M} \quad \qedhere
    \end{align*}
\end{proof}

%\begin{theorem}[Tightness of Random Precoding Baseline]
%    \label{thm:randomprecodingtight}
%    For a general distribution $\mathcal H$ of channel vectors $\mathbf h_\mathrm{D} \in \mathbb C^M, ~ \mathbf h_\mathrm{D} \sim \mathcal H$, it is impossible to find a distribution of precoding vectors $\mathbf w$ with $\mathbf w, \mathbf h_\mathrm{D}$ independent and $\lVert \mathbf w \rVert^2 = 1$ such that
%    \begin{equation}
%        \bar P_\mathrm{indep} := \mathrm{E} \left[ \frac{|\mathbf h_\mathrm{D}^\mathrm{H} \mathbf w|^2}{\lVert \mathbf h_\mathrm{D} \rVert^2} \right] > \frac{1}{M}.
%        \label{eq:randomprecodingtight}
%    \end{equation}
%\end{theorem}
%\begin{proof}
%    Since Theorem \ref{thm:randomprecodingtight} must hold for general distributions $\mathcal H$, in particular, it must hold for $\mathbf h_\mathrm{D} \sim \mathcal {CN}(\mathbf 0, \mathbf I_M)$.
%    A proof similar to the one for \ref{thm:randomprecoding} for this choice of $\mathcal H$ reveals that $\bar P_\mathrm{indep} = \frac{1}{M}$, which contradicts Eq. \ref{eq:randomprecodingtight}.
%\end{proof}

%Theorem \ref{thm:randomprecodingtight} indicates that if no a-priori information about the distribution of channel vectors $\mathbf h_\mathrm{D}$ is available, random precoding is indeed the best precoding strategy.
For real channels, random precoding is not a fair benchmark to compare \ac{NN}-generated estimates against, since it does not take the prior distribution of $\mathbf h_\mathrm{D}$ over a dataset into account.
A badly designed \ac{NN} could just learn the prior distribution of $\mathbf h_\mathrm{D}$ and not extract information from $\mathbf H_\mathrm{U}$.
As another baseline, Theorem \ref{thm:apriori} describes a precoding technique with a constant \ac{DL} channel estimate $\mathbf w$ that exploits a-priori information.

\begin{theorem}[Principal Component Baseline]
    \label{thm:apriori}
    To maximize the mean normalized power $\bar P$ over the distribution $(\mathbf H_\mathrm{U}, \mathbf h_\mathrm{D}) \sim \mathcal H$ under the restriction that the \ac{DL} channel estimate $\mathbf w$ is constant and $\lVert \mathbf w \rVert = 1$, $\mathbf w$ must be chosen such that $\mathbf w = \mathbf w_\mathrm{max}$, where $\mathbf w_\mathrm{max}$ is the eigenvector corresponding to the largest eigenvalue of the auto-correlation matrix $\mathbf R:= \mathrm{E}_{\mathbf h_\mathrm{D}} \left[ \frac{\mathbf h_\mathrm{D} \mathbf h_\mathrm{D}^\mathrm{H}}{|\mathbf h_\mathrm{D}^\mathrm{H} \mathbf h_\mathrm{D}|} \right]$. We define
    \[
        \bar P_\mathrm{princ} := \max_{\lVert\mathbf w\rVert = 1} ~ \mathrm{E}_{\mathbf h_\mathrm{D}} \left[ \frac{|\mathbf h_\mathrm{D}^\mathrm{H} \mathbf w|^2}{\lVert \mathbf h_\mathrm{D} \rVert^2} \right].
    \]
\end{theorem}
\begin{proof}
    The objective is to find $\mathbf w_\mathrm{max}$ according to
    \begin{align*}
        \mathbf w_\mathrm{max} &=% \argmax_{\lVert\mathbf w\rVert = 1} ~ \mathrm{E}_{\mathbf h_\mathrm{D}} \left[ \frac{|\mathbf h_\mathrm{D}^\mathrm{H} \mathbf w|^2}{\lVert \mathbf h_\mathrm{D} \rVert^2} \right] \\
        \argmax_{\lVert\mathbf w\rVert = 1} ~ \mathrm{E}_{\mathbf h_\mathrm{D}} \left[ \frac{\mathbf w^\mathrm{H} \mathbf h_\mathrm{D} \mathbf h_\mathrm{D}^\mathrm{H} \mathbf w}{\lVert \mathbf h_\mathrm{D} \rVert^2} \right] \\
        &= \argmax_{\lVert\mathbf w\rVert = 1} ~ \mathbf w^\mathrm{H} \mathbf R \mathbf w.
    \end{align*}
    By taking the derivative of the Lagrange function $\mathcal L(\mathbf w) = \mathbf w^{\mathrm H} \mathbf R \mathbf w - \lambda (\mathbf w^{\mathrm H} \mathbf w - 1)$ with respect to $\mathbf w$, we find that $\mathbf R \mathbf w = \lambda \mathbf w$.
    Hence, $\mathbf w$ is an eigenvector of $\mathbf R$ and the function in
    \[
        \mathbf w_\mathrm{max} = \argmax_{\lVert\mathbf w\rVert = 1} ~ \mathbf w^\mathrm{H} \mathbf R \mathbf w = \argmax_{\lVert\mathbf w\rVert = 1} ~ \mathbf w^\mathrm{H} \lambda \mathbf w
    \]
    is maximized if $\mathbf w$ corresponds to the largest eigenvalue $\lambda$.
\end{proof}

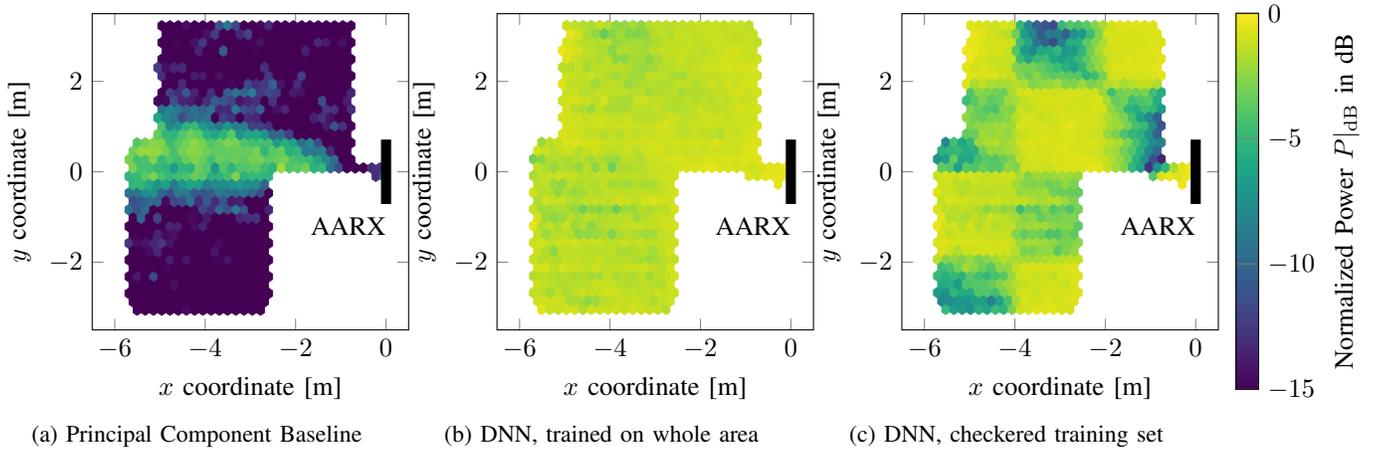
\begin{figure*}
    \begin{subfigure}[t]{0.29\textwidth}
        \centering
        \begin{tikzpicture}
            \begin{axis}[
                width=0.8\textwidth,
                height=0.8\textwidth,
                scale only axis,
                xmin=-6.5,
                xmax=0.5,
                ymin=-3.5,
                ymax=3.5,
                xlabel={$x$ coordinate [m]},
                ylabel={$y$ coordinate [m]}
                ]
        
                \addplot[thick,blue] graphics[xmin=-6.5,ymin=-3.5,xmax=0.5,ymax=3.5] {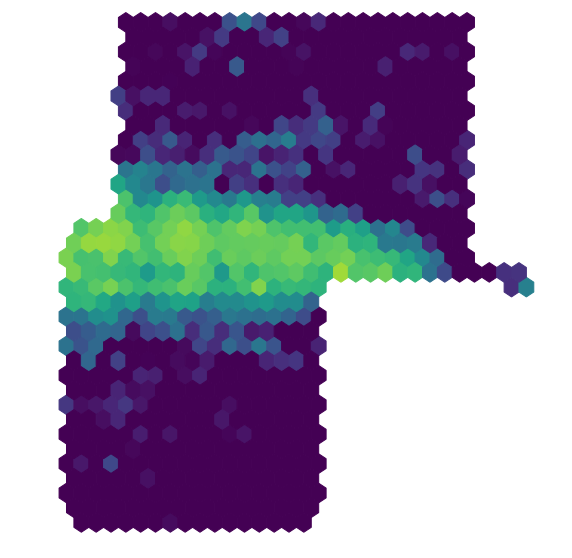};
                \draw [fill] (axis cs: -0.1, -0.7) rectangle (axis cs: 0.1, 0.7);
                \node [anchor = north east] at (axis cs: 0.2, -0.8) {AARX};
            \end{axis}
        \end{tikzpicture}
        \caption{Principal Component Baseline}
        \label{fig:heatmapprinc}
    \end{subfigure}
    \begin{subfigure}[t]{0.29\textwidth}
        \centering
        \begin{tikzpicture}
            \begin{axis}[
                width=0.8\textwidth,
                height=0.8\textwidth,
                scale only axis,
                xmin=-6.5,
                xmax=0.5,
                ymin=-3.5,
                ymax=3.5,
                xlabel={$x$ coordinate [m]},
                ylabel={$y$ coordinate [m]}
                ]
        
                \addplot[thick,blue] graphics[xmin=-6.5,ymin=-3.5,xmax=0.5,ymax=3.5] {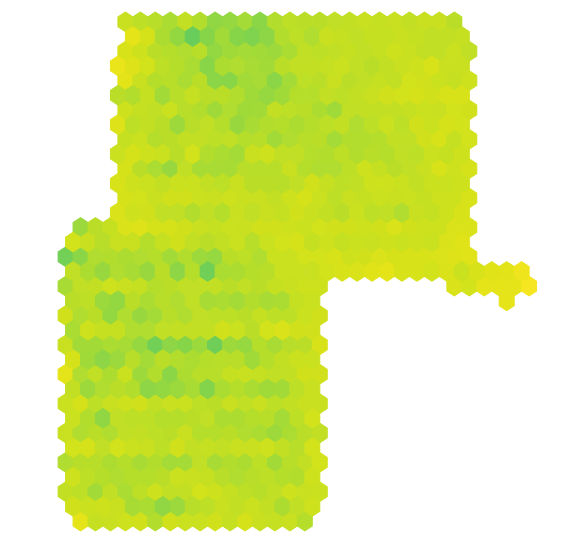};
                \draw [fill] (axis cs: -0.1, -0.7) rectangle (axis cs: 0.1, 0.7);
                \node [anchor = north east] at (axis cs: 0.2, -0.8) {AARX};
            \end{axis}
        \end{tikzpicture}
        \caption{\ac{DNN}, trained on whole area}
        \label{fig:heatmapdnn}
    \end{subfigure}
    \begin{subfigure}[t]{0.29\textwidth}
        \centering
        \begin{tikzpicture}
            \begin{axis}[
                width=0.8\textwidth,
                height=0.8\textwidth,
                scale only axis,
                xmin=-6.5,
                xmax=0.5,
                ymin=-3.5,
                ymax=3.5,
                xlabel={$x$ coordinate [m]},
                ylabel={$y$ coordinate [m]}
                ]
                
                \addplot[thick,blue] graphics[xmin=-6.5,ymin=-3.5,xmax=0.5,ymax=3.5] {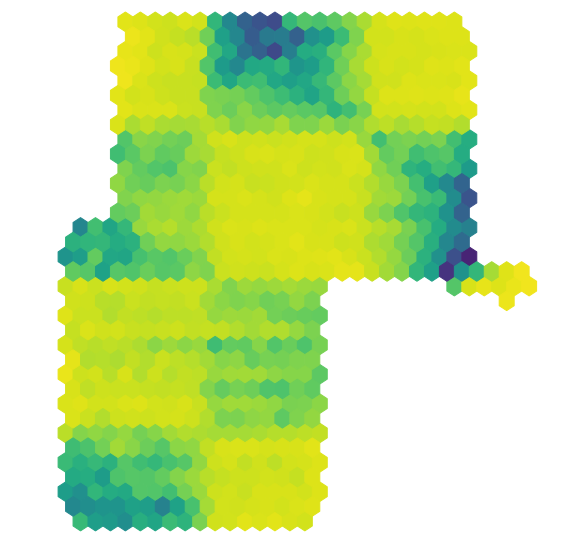};
                \draw [fill] (axis cs: -0.1, -0.7) rectangle (axis cs: 0.1, 0.7);
                \node [anchor = north east] at (axis cs: 0.2, -0.8) {AARX};
            \end{axis}
        \end{tikzpicture}
        \caption{\ac{DNN}, checkered training set}
        \label{fig:heatmapcheckered}
    \end{subfigure}
    \begin{subfigure}[b]{0.09\textwidth}
        \begin{tikzpicture}
            \begin{axis}[
                hide axis,
                scale only axis,
                colormap/viridis,
                colorbar,
                point meta min=-15,
                point meta max=0,
                colorbar style={
                    height=5cm,
                    width=0.3cm,
                    ytick={-15,-10,-5,0},
                    ylabel={Normalized Power $P|_\mathrm{dB}$ in dB}
                }]
                \addplot [draw=none] coordinates {(0,0)};
            \end{axis}
        \end{tikzpicture}
    \end{subfigure}
    \caption{Top view of normalized received powers $P$ at different locations visualized over the approximately $6\,\mathrm{m} \times 6\,\mathrm{m}$ large measurement area in the dataset. The black box marked ``AARX'' indicates the location of the antenna array.}
\end{figure*}

\section{Measurement Dataset}
\label{sec:dataset}

For evaluating our deep learning-based \ac{CSI} estimation, we draw on a dataset measured with our own channel sounder called \ac{DICHASUS} \cite{dichasus2021}.
More specifically, we use a publicly available indoor dataset entitled \emph{dichasus-015x} measured with an $M = 32$-antenna uniform planar array in a $6\,\mathrm{m} \times 6\,\mathrm{m}$ office room \cite{dataset-dichasus-015x}.
Overall, the dataset contains more than $85\,000$ position-tagged \ac{CSI} datapoints captured at a carrier frequency of $f_\mathrm{c} = 1.272\,\mathrm{GHz}$.
Each \ac{CSI} datapoint was estimated from multiple \ac{OFDM} symbols with $N_\mathrm{sub} = 1024$ subcarriers spread over a bandwidth of $50\,\mathrm{MHz}$.
We averaged over batches of 32 neighboring subcarriers for the purpose of \ac{DL} channel estimation, obtaining a total of $32$ averaged channel coefficients.
% LoS, reflected

From this large $50\,\mathrm{MHz}$ bandwidth, we collect channel coefficients within some ranges into \emph{virtual} uplink and downlink channels.
As shown in Fig. \ref{fig:principle}, we grouped the channel coefficients for averaged subcarriers 0-7 to be the virtual \ac{UL} channel ($\mathbf H_\mathrm{U}$) and we call the channel coefficients for subcarrier 28 the virtual \ac{DL} channel vector ($\mathbf h_\mathrm{D}$).
Note that the \ac{CSI} dataset was measured with all antennas in the array exclusively operated as receivers at carrier frequency $f_\mathrm{c}$, but, thanks to channel reciprocity, the same channel coefficients can be assumed for the \ac{DL} direction.
Our choice corresponds to a virtual \ac{UL} channel with a bandwidth of $12.5\,\mathrm{MHz}$ centered around $f_\mathrm{c, UL} \approx 1.2533\,\mathrm{GHz}$ and a virtual \ac{DL} channel coefficient measured at carrier frequency $f_\mathrm{c, DL} \approx 1.2915\,\mathrm{GHz}$.
The center frequencies of uplink channel and downlink subcarrier are separated by $f_\mathrm{c, DL} - f_\mathrm{c, UL} \approx 38.2\,\mathrm{MHz}$.

Through random precoding according to Theorem \ref{thm:randomprecoding}, it is always possible to achieve a mean received power of $\bar P_\mathrm{rand} = \frac{1}{32}, ~ \bar P_\mathrm{rand}|_\mathrm{dB} \approx -15\,\mathrm{dB}$, i.e., approximately $15\,\mathrm{dB}$ less on average than is possible if the true channel vector $\mathbf h_\mathrm{D}$ was known by the \ac{BS}.
When precoding with the optimal constant \ac{DL} channel estimate $\mathbf w_\mathrm{max}$ according to Theorem \ref{thm:apriori}, we find that it is possible to achieve $\bar P_\mathrm{princ}|_\mathrm{dB} \approx -8.8\,\mathrm{dB}$ just by exploiting the prior distribution of the dataset.
The distribution of received powers over the dataset's measurement area for this case is illustrated in Fig. \ref{fig:heatmapprinc}:
Precoding with $\mathbf w_\mathrm{max}$ generates a single broad, forward-facing beam. 

\section{Deep Learning-Based CSI Estimation}
\label{sec:deeplearning}
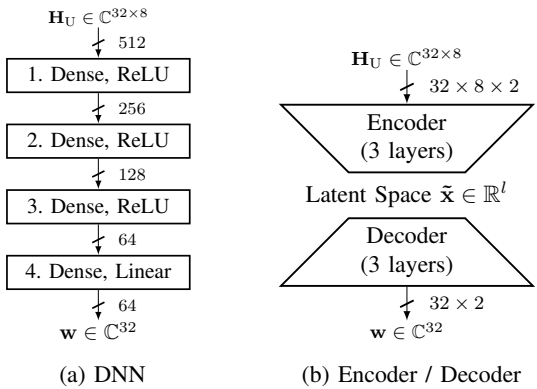
\begin{figure}
    \centering
    \begin{subfigure}[b]{0.45\columnwidth}
        \centering
        \scalebox{0.9}{
            \scalebox{0.9}{
    \begin{tikzpicture}
    	%\node (l1) [minimum width = 3cm, minimum height = 0.55cm, draw, thick, inner sep = 3pt] at (0, 0) {Flatten};
    	\node (l1) [minimum width = 3cm, minimum height = 0.55cm, draw, thick, inner sep = 3pt] at (0, 0) {1. Dense, ReLU};
    	\node (l2) [minimum width = 3cm, minimum height = 0.55cm, draw, thick, inner sep = 3pt, below = 0.5cm of l1] {2. Dense, ReLU};
    	\node (l3) [minimum width = 3cm, minimum height = 0.55cm, draw, thick, inner sep = 3pt, below = 0.5cm of l2] {3. Dense, ReLU};
    	\node (l4) [minimum width = 3cm, minimum height = 0.55cm, draw, thick, inner sep = 3pt, below = 0.5cm of l3] {4. Dense, Linear};
    	%\node (l6) [minimum width = 3cm, minimum height = 0.55cm, draw, thick, inner sep = 3pt, below = 0.5cm of l5] {Reshape};
    
    	\node [anchor = south] at ($(l1.north) + (0, 0.4)$) {\footnotesize $\mathbf H_\mathrm{U} \in \mathbb C^{32 \times 8}$};
    	\node [anchor = north] at ($(l4.south) + (0, -0.4)$) {$\mathbf w \in \mathbb C^{32}$};

    	\draw [-latex] ($(l1.north) + (0, 0.5)$) -- (l1) node[midway, inner sep = 0] (p0to1) {};
    	\draw [-latex] (l1) -- (l2) node[midway, inner sep = 0] (p1to2) {};
    	\draw [-latex] (l2) -- (l3) node[midway, inner sep = 0] (p2to3) {};
    	\draw [-latex] (l3) -- (l4) node[midway, inner sep = 0] (p3to4) {};
    	%\draw [-latex] (l4) -- (l5) node[midway, inner sep = 0] (p4to5) {};
    	%\draw [-latex] (l5) -- (l6) node[midway, inner sep = 0] (p5to6) {};
    	\draw [-latex] (l4) -- ($(l4.south) + (0, -0.5)$) node[midway, inner sep = 0] (p4to5) {};

    	%\draw [thick] ($(p0to1) + (-0.1, -0.05)$) -- ($(p0to1) + (+0.1, +0.05)$) node[midway, anchor = west, xshift = 0.2cm] {\footnotesize $32 \times 8 \times 2$};	
    	\draw [thick] ($(p0to1) + (-0.1, -0.05)$) -- ($(p0to1) + (+0.1, +0.05)$) node[midway, anchor = west, xshift = 0.2cm] {\footnotesize $512$};
    	\draw [thick] ($(p1to2) + (-0.1, -0.05)$) -- ($(p1to2) + (+0.1, +0.05)$) node[midway, anchor = west, xshift = 0.2cm] {\footnotesize $256$};
    	\draw [thick] ($(p2to3) + (-0.1, -0.05)$) -- ($(p2to3) + (+0.1, +0.05)$) node[midway, anchor = west, xshift = 0.2cm] {\footnotesize $128$};
    	\draw [thick] ($(p3to4) + (-0.1, -0.05)$) -- ($(p3to4) + (+0.1, +0.05)$) node[midway, anchor = west, xshift = 0.2cm] {\footnotesize $64$};
    	\draw [thick] ($(p4to5) + (-0.1, -0.05)$) -- ($(p4to5) + (+0.1, +0.05)$) node[midway, anchor = west, xshift = 0.2cm] {\footnotesize $64$};
    	%\draw [thick] ($(p6to7) + (-0.1, -0.05)$) -- ($(p6to7) + (+0.1, +0.05)$) node[midway, anchor = west, xshift = 0.2cm] {\footnotesize $32 \times 2$};

    % 	\node [red] at (0, 0) {\huge TODO};
    \end{tikzpicture}
}
        }
        \caption{\ac{DNN}}
        \label{fig:autoencoder}
    \end{subfigure}
    \begin{subfigure}[b]{0.45\columnwidth}
        \centering
        \scalebox{0.9}{
    
    \begin{tikzpicture}
        
        \node (l1) [trapezium, draw, thick, align=center, trapezium angle=315, inner sep=2pt, minimum height=1cm]{Encoder\\(3 layers)};
        
        \node (l2) [align=center, below=0cm of l1]{Latent Space $\mathbf {\tilde x} \in \mathbb R^l$};

        \node (l3) [trapezium, draw, thick, align=center, trapezium angle=45, inner sep=2pt, below=0.65cm of l1, minimum height=1cm, minimum width=2cm]{Decoder\\(3 layers)};

        \node [anchor = south] at ($(l1.north) + (0, 0.4)$) {\footnotesize $\mathbf H_\mathrm{U} \in \mathbb C^{32 \times 8}$};
    	\node [anchor = north] at ($(l3.south) + (0, -0.4)$) {\footnotesize $\mathbf w \in \mathbb C^{32}$};

        \draw [-latex] ($(l1.north) + (0, 0.5)$) -- (l1) node[midway, inner sep = 0] (p0to1) {};
        \draw [-latex] (l3) -- ($(l3.south) + (0, -0.5)$) node[midway, inner sep = 0] (p3to4) {};

        \draw [thick] ($(p0to1) + (-0.1, -0.05)$) -- ($(p0to1) + (+0.1, +0.05)$) node[midway, anchor = west, xshift = 0.2cm] {\footnotesize $32 \times 8 \times 2$};	
        \draw [thick] ($(p3to4) + (-0.1, -0.05)$) -- ($(p3to4) + (+0.1, +0.05)$) node[midway, anchor = west, xshift = 0.2cm] {\footnotesize $32 \times 2$};

    \end{tikzpicture}
}
        \caption{Encoder / Decoder}
        \label{fig:nnarchitecture}
    \end{subfigure}
    \caption{Architecture of different evaluated \acp{NN}}
\end{figure}

We evaluate five different deep learning-based downlink \ac{CSI} estimators $\hat { \bm \theta }$ which produce an estimate $\mathbf w$ from $\mathbf H_\mathrm{U}$:
\begin{itemize}
    \item \textbf{A \ac{DNN}}: This simple architecture consists of four dense layers as shown in Fig. \ref{fig:nnarchitecture}.
    \item \textbf{A \ac{DNN} with dropout}: Same as the \ac{DNN} architecture, except for a dropout layer with dropout rate $\delta$ inserted between dense layers 2 and 3, to improve generalization.
    \item \textbf{An Encoder / Decoder structure} with arbitrary latent space: both encoder and decoder consist of three dense hidden layers each. The encoder reduces $\mathbf H_\mathrm{UL}$ to a latent space representation $\mathbf {\tilde x} \in \mathbb R^l$, that the decoder infers $\mathbf h_\mathrm{D}$ from, see Fig. \ref{fig:autoencoder}. This choice of network architecture is justified in the fact that, as explained in Section \ref{sec:baselines}, both $\mathbf H_\mathrm{U}$ and $\mathbf h_\mathrm{D}$ are entirely predetermined by a possibly sparser latent representation $\mathbf x$. The Encoder may be able to approximate $f_\mathrm{U}^{-1}$ whereas the decoder may approximate $f_\mathrm{D}$.
    \item \textbf{An azimuth angle-based Encoder / Decoder structure}, i.e., the latent variable is forced to be an azimuth angle: Same as the previous architecture, except that encoder and decoder are now first trained separately: The encoder is trained to estimate the azimuth component $\alpha$ of the \ac{AoA} from $\mathbf H_\mathrm{U}$ and the decoder is trained to generate $\mathbf h_\mathrm{D}$ from $\alpha$, both supervised using position labels. The two \acp{NN} are then connected in series.
    \item \textbf{An azimuth and elevation angle-based Encoder / Decoder structure}, i.e., the two latent variables are forced to be elevation / azimuth angle estimates: Same as the previous architecture, except that the encoder now consists of two separate \acp{DNN}, for estimating both azimuth component $\alpha$ and elevation component $\beta$ of the \ac{AoA}. Again, the decoder is not trained on estimates, but on \acp{AoA} computed from position labels.
\end{itemize}

\begin{figure*}
    \centering
    \scalebox{0.95}{
        \input{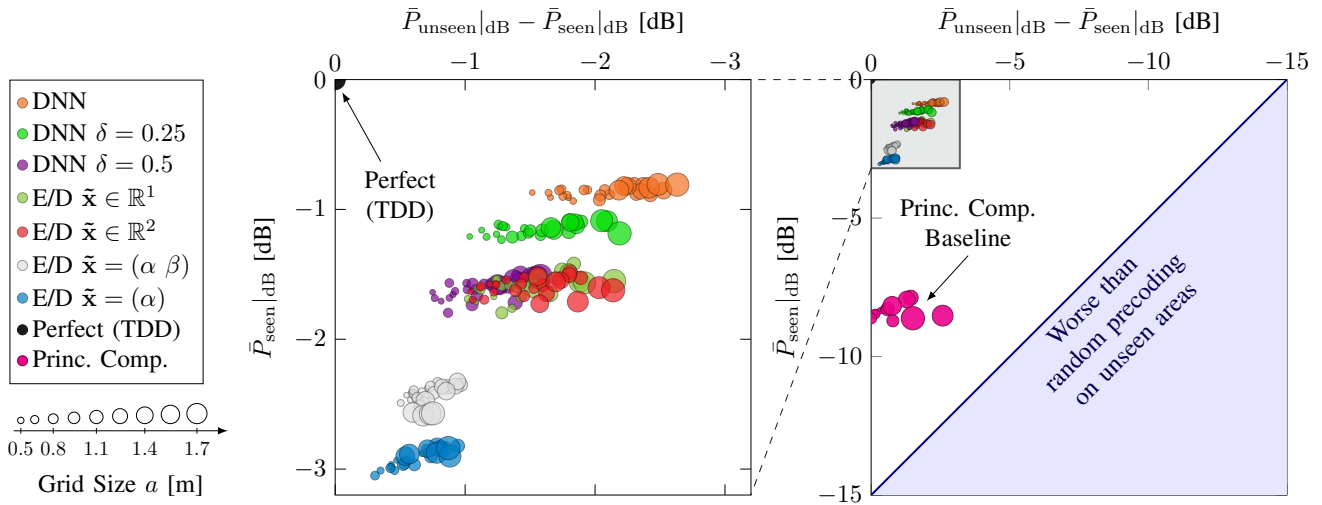}
    }
    \caption{\emph{Seen/unseen loss} diagram with mean received power losses of different \ac{NN} architectures and baselines on seen / unseen checkerboard fields, for grid sizes from $0.5\,\mathrm m$ to $1.8\,\mathrm m$. A larger grid size is indicated by a larger marker size.}
    \label{fig:generalization}
\end{figure*}

As a training loss function, we employ $\ell = 1 - P$, where $P$ is the squared cosine similarity between estimated channel $\mathbf w$ and true channel $\mathbf h_\mathrm{D}$ as defined in Eq. (\ref{eq:power}).
Instead of working with complex-valued channel coefficients, all \acp{NN} process channel coefficients in real / imaginary part representation.
For Fig. \ref{fig:heatmapdnn}, we randomly assigned $50\,\%$ of all datapoints to the training set, trained the previously described \ac{DNN} (without dropout) on this set and evaluated the \ac{DL} \ac{CSI} estimates on the complete dataset.
The mean normalized received power over the complete dataset was found to be $\bar P|_\mathrm{dB} \approx -1.3\,\mathrm{dB}$.

For Fig. \ref{fig:heatmapcheckered}, on the other hand, we partitioned the dataset into training set and test set in a checkerboard pattern with square side length $2\,\mathrm m$: All datapoints that were measured on ``white'' checkerboard squares were assigned to the training set, all datapoints measured on ``black'' checkerboard squares to the test set.
After evaluating the trained \ac{NN} on both training and test set, it is easy to see that the performance on the two sets differs significantly, with $\bar P|_\mathrm{dB} \approx -0.9\,\mathrm{dB}$ on the training set and $\bar P|_\mathrm{dB} \approx -4.2\,\mathrm{dB}$ on the test set.
Clearly, this indicates that the \ac{DNN} is overfitting on the training set:
It is not able to produce channel estimates $\mathbf w$ with comparable normalized \ac{DL} power on unseen regions of the physical space.

This poses the question as to how this overfitting can be mitigated, either through standard methods such as adding dropout layers or by forcing the \ac{DNN} to learn a sparser latent space representation.
To formalize our enquiry into this topic into a quantifiable manner, we introduce a framework for evaluating the ability to generalize of different \ac{NN} architectures.

\section{A Framework for Evaluating Generalization}
\label{sec:generalization}

\subsection{Defining Generalization}
In the context of \ac{DL} \ac{CSI} estimation, we refer to generalization not as the ability to generalize from training set to test set in the same physical area (which a \ac{DNN} can do well on our dataset, as is apparent from Fig. \ref{fig:heatmapdnn}), but as the ability to generalize from areas of the physical environment seen during training to areas that were not represented in the training set (which, considering the result in Fig. \ref{fig:heatmapcheckered}, is much harder).
Therefore, when talking about the quality of \ac{CSI} estimates, it is insufficient to just measure a single performance metric: Some \ac{NN} architectures perform well in seen areas, but worse in unseen locations whereas other architectures generalize better, but at the cost of a worse performance in seen areas.

To quantify this observation, as previously, the dataset is split into training and test set in a checkerboard pattern, with square side length $a$.
If $a$ is chosen to be small, the training set will contain a \ac{CSI} datapoint in physical proximity of each (unseen) test set location.
For large values of $a$, the \ac{NN} needs to be able to generalize across larger distances.
We define $\bar P_\mathrm{seen}$ to be the average received power (see Eq. (\ref{eq:poweravg})) when evaluating the trained \ac{NN} on the training set, and $\bar P_\mathrm{unseen}$ to be the average received power after evaluation on the test set.
$\bar P_\mathrm{seen}$ can be interpreted as the average loss in received power due to the suboptimal channel coefficient estimates.
We expect $\bar P_\mathrm{seen} \geq \bar P_\mathrm{unseen}$, so $\bar P_\mathrm{unseen}|_\mathrm{dB} - \bar P_\mathrm{seen}|_\mathrm{dB} < 0\,\mathrm{dB}$ can be interpreted as the loss in average received power incurred in unseen areas due to lack of training data in physical proximity.

\subsection{Seen/Unseen Loss Diagram and Baselines}
To visualize \ac{NN} performance, we propose a \emph{seen/unseen loss} diagram as in Fig. \ref{fig:generalization}, with losses $\bar P_\mathrm{seen}|_\mathrm{dB}$ on the horizontal axis and $\bar P_\mathrm{unseen}|_\mathrm{dB} - \bar P_\mathrm{seen}|_\mathrm{dB}$ on the vertical axis.
In any case, the random precoding strategy from Thm. \ref{thm:randomprecoding} provides a lower bound on the achievable performance (blue line and region).
The best performance is achieved if perfect \ac{DL} \ac{CSI} is available at the receiver at all time, so that $\bar P_\mathrm{seen}|_\mathrm{dB} = \bar P_\mathrm{unseen}|_\mathrm{dB} = 0\,\mathrm{dB}$; this operating point is marked with ``\ac{TDD}'', since, assuming perfect channel reciprocity, it is achievable by a \ac{TDD} system.
For all other estimators, the performance in seen and unseen areas depends on the partitioning of the dataset into training and test set.
For Fig. \ref{fig:generalization}, this partitioning was performed in the afforementioned checkerboard pattern.
The grid size parameter $a$ was swept from $0.5\, \mathrm m$ to $1.8\,\mathrm m$ with a step size of $0.1\,\mathrm m$.
As an additional baseline, based on Thm. \ref{thm:apriori}, we compute $\mathbf w_\mathrm{max}$ based on the training set and evaluate this vector for both training set ($\bar P_\mathrm{seen}$) and test set ($\bar P_\mathrm{unseen}$), yielding the principal component baseline (marked ``Princ. Comp.'') also illustrated in Fig. \ref{fig:generalization}.

\subsection{Discussion of Results}
Among all tested \acp{NN}, the \ac{DNN} without dropout performs best on previously seen data (i.e., with respect to $\bar P_\mathrm{seen}|_\mathrm{dB}$).
Increasing the dropout rate $\delta$ to $\delta = 0.25$ or $\delta = 0.5$ leads to a deteriorated performance with respect to $\bar P_\mathrm{seen}$, but better generalization.
Surprisingly, both encoder / decoder structures without predetermined latent space perform approximately equally well, regardless of the latent space dimensionality ($\mathbb R^1$ or $\mathbb R^2$), which may indicate that a sparse representation of \ac{CSI} is indeed possible.
A closer look at the learned latent representation would reveal that $\mathbf {\tilde x} \in \mathbb R^1$ is highly correlated with the azimuth angle.
Despite this observation, encoder / decoder structures with predetermined azimuth $\alpha$ / elevation $\beta$ latent spaces perform worse than all other \ac{NN} architectures on previously seen physical areas, but generalize better.

We find that the performance of all evaluated \ac{NN} architectures is significantly better than both \emph{random precoding} and \emph{principal component} baselines.
In fact, $\bar P_\mathrm{seen}|_\mathrm{dB} > -3.1\,\mathrm{dB}$ and $\bar P_\mathrm{unseen}|_\mathrm{dB} > -3.8\,\mathrm{dB}$ for all \acp{NN}, which demonstrates that a \ac{NN}-based approach is feasible and that some level of generalization to previously unseen physical areas is possible.
However, Fig. \ref{fig:generalization} also clearly shows significant performance differences between the various \ac{NN} architectures and the strong influence of the grid size on generalization.

\section{Summary and Outlook}
We found that \ac{NN}-based downlink channel estimation from available uplink \ac{CSI} significantly outperformed the baselines and that generalization to physical areas not represented in the training set is one of the major challenges of the approach.
With regards to generalization, we evaluated several different network architectures on measurement data.
Thanks to the public data, our research may be reproduced on the same dataset or compared to other datasets captured in different types of environments or with different antenna configurations.
The effect of the frequency separation between uplink and downlink channel may also be studied further.

\bibliographystyle{IEEEtran}
\bibliography{IEEEabrv,references}

\end{document}